\documentclass[12pt]{gtpart}

\usepackage{xypic}
\usepackage{graphicx}
\usepackage{amsfonts}
\usepackage{mathrsfs}
\usepackage{amssymb}
\usepackage{amsxtra}
\def\G{\Gamma}
\def\FF{\mathbb F}
\def\del{\partial}
\def\SS{\mathbb S}
\def\Z{\mathbb Z}
\def\T{\mathbb T}
\def\R{\mathbb R}
\def\C{\mathbb C}
\def\Q{\mathbb Q}
\def\TTheta{{\mathbb T}_{\Theta}}

\def\sT{\mathscr T}
\def\B{\mathscr B}
\def\bG{\bar \Gamma}

\def\t{\tau}

\def\H{\mathscr{H}}

\def\Gpd{\mathcal G}

\def\Tdeg{T_{\rm deg}}

\newtheorem{thm}{Theorem}[section]

\newtheorem{prop}[thm]{Proposition}
\newtheorem{cor}[thm]{Corollary}

\theoremstyle{definition}

\newtheorem{df}[thm]{Definition}

\newtheorem{rmk}[thm]{Remark}

\def\BTheta{\B_{\Theta}}

\newcommand{\cev}[1]{\stackrel{\leftarrow}{#1}}
\renewcommand{\vec}[1]{\stackrel{\rightarrow}{#1}}
\accentedsymbol{\dbarG}{\Bar{\Bar{\Gamma}}}

\begin{document}

\title[Geometry of the momentum space]
{Geometry of the momentum space: From wire networks to quivers and monopoles}

\author
[Ralph M.\ Kaufmann]{Ralph M.\ Kaufmann}
\email{rkaufman@math.purdue.edu}

\address{Department of Mathematics, Purdue University,
 West Lafayette, IN 47907}

\author
[Sergei Khlebnikov]{Sergei Khlebnikov}
\email{skhleb@physics.purdue.edu}

\address{Department of  Physics, Purdue University,
 West Lafayette, IN 47907}

\author
[Birgit Kaufmann]{Birgit Wehefritz--Kaufmann}
\email{ebkaufma@math.purdue.edu}

\address{Department of Mathematics and Department of Physics, Purdue University, West Lafayette, IN 47907}

\begin{abstract} A new nano--material in the form of a double gyroid
has motivated us to study (non)--commutative $C^*$ geometry
of periodic wire networks and the associated graph Hamiltonians. Here
we present the general abstract framework, which is given by certain quiver representations,
with special attention to the original case of the gyroid as well as related cases,
such as  graphene.
In these geometric situations, the non--commutativity is introduced by a constant magnetic field
and the theory splits into two pieces: commutative and non--commutative, both of which are governed by a $C^*$ geometry.

In the non--commutative case, we can use tools such as K--theory to make statements about the band structure.
In the commutative case,
we give geometric and algebraic methods to study band intersections;
these methods come from singularity theory and
representation theory. We also provide new tools in the study, using $K$--theory and Chern classes. The latter can be computed using Berry connection in the momentum space.
This brings monopole charges and issues of topological stability into the picture.

\end{abstract}

\maketitle


\section{Introduction}
Recently, a new nano--material in the form of a double gyroid has been synthesized
\cite{Hillhouse}. It is based on a thickened triply-periodic minimal surface, whose
complement consists of two non-intersecting channels. These can be filled with conducting
or semiconducting materials \cite{Hillhouse}
to function as nanowire networks with potentially useful
electronic properties \cite{Khlebnikov&Hillhouse}.
The nontrivial topology of such a network has motivated our study
of its commutative and non--commutative geometry \cite{kkwk}. Following
Bellissard and Connes \cite{B,Connes, MM}, we proceed by indentifying the relevant
$C^*$--algebra, which in our case is spanned by the symmetries and the tight-binding (Harper)
Hamiltonian of the skeletal graph obtained
as a deformation retract of the channel; we call it the Bellissard-Harper algebra.
This approach leads to an effective  geometry described by a family of finite dimensional
Hamiltonians and their spectra; the latter determine the band structure of the
original nanostructured solid in the tight-binding approximation.
By placing the material into an external constant magnetic field the geometry is rendered
noncommutative.

In this paper, we generalize that setup to further noncommutative geometries obtained via  certain quiver  representations.
We also adapt the techniques of \cite{kkwk3} and \cite{kkwk4} to this more general situation. In particular, in the commutative case, we get a classification of singularities in
the spectrum---the band intersections.
The simplest of these is a conical intersection of two bands, commonly referred to as a Dirac
point. We give analytic tools to compute locations and properties of the singular points.

In the general framework above, we also give a new interpretation of the Berry phase phenomenon \cite{Berry} in terms of $K$--theory and Chern--classes
generalizing the observations of Thouless {\em et al.} (TKNN) \cite{TKNN} and Simon
\cite{simon}.
These concepts include topological charges in various guises: scalar, $K$--theoretic
and cohomological.
When the parameter space is three-dimensional,
isolated conical degeneracies are magnetic monopoles in the parameter space \cite{Berry}.
In the present case, the parameters are components of the crystal momentum ${\bf k}$;
their number equals the dimensionality
of the original periodic structure.
Thus, in three spatial dimensions---the case of the gyroid---Dirac
points are monopoles in the momentum space and, as we will see, are stable with respect to
small deformations of the graph Hamiltonian.
Furthermore, using foliations, we consider a slicing technique which leads to an
effective numerical tool for finding singular points in the spectrum, generalizing the
method used for this purpose in \cite{Xu&al}. This technique has been implemented
in \cite{kkwk5} and corroborates the topological stability of the gyroid's
Dirac points.
This stability is not a common characteristic of
all Dirac points: those of graphene, which is described by the honeycomb lattice, do not
exhibit this property, see e.g.\ \cite{Fefferman}\footnote{That article also
explores deformation direction where the Dirac points do stay stable.}.
This fact has an elegant and short explanation in our approach.
We expect that this analysis will contribute to understanding of
potential applications of gyroid-based nanomaterials, as well as to the theory of
three-dimensional generalizations of the quantum Hall effect, along the lines of \cite{BE}.
In two dimensions, the TKNN equations for generalized Dirac--Harper operators have
been worked out in \cite{Landi}.
Analyses of higher-dimensional situations are contained in
\cite{Demikhovskii, Goldman, Bernevig, Koshino, Goryo}.

Even without going to complete generality provided by quiver representations, our approach
to studying wire networks is not restricted to the gyroid system
and applies to any embedded periodic wire network in $\R^n$.
We have already used it to study more examples, namely, Bravais lattices, the honeycomb lattice and two other triply periodic surfaces and their wire networks, the primitive cubic (P surface) and the diamond (D surface). We refer to these as the geometric examples.
We recall some results here and include a new consideration of the topological charges.
In  these cases the noncommutative geometry is given by
a subalgebra of a matrix algebra with coefficients in the noncommutative torus. Here the parameters of the torus correspond to the $B$--field that the material is subjected to.

One surprising fact is that  some properties of the non--commutative situation are similar  to the situation without a magnetic field, and there is evidence for duality
 between these two situations.
The duality concerns the degenerate subspaces of the torus that appears as the
relevant moduli space in both cases.
In the commutative case, i.e. in the absence of a magnetic field, the torus is the base for the family of Hamiltonians and the requsite subspace is where the spectrum of the Hamiltonian has degeneracies. In the noncommutative case, the same torus parameterizes the $B$--field and the locus of degeneracy is that of those values of $B$ where the Bellissard--Harper algebra is not the full matrix algebra.

The paper is organized as follows: we start with a description of the material and its underlying geometry in Chapter 2. Here the geometry is reduced to that of
the skeletal graph---the deformation retract of a channel component of the complement to the triply periodic surface. We also introduce other related geometries which we consider in parallel. These are the honeycomb lattice underlying graphene, and the P and D surfaces, which are the other triply periodic self--symmetric  surfaces. Chapter 3 describes the mathematical model we work with. This includes the Harper Hamiltonian and the relevant Hilbert space and $C^*$ algebra, the Bellissard--Harper algebra.  We discuss  the $C^*$ geometry in Chapter 4. This includes our analysis of the Berry connection, topological charges and stability
of the singular points as well as a slicing method to detect singular points or monopoles. Chapter 5 contains our results about degeneracies in the spectrum of the Harper Hamiltonian in the commutative case using singularity and representation theories. In Chapter 6 we summarize the results of our analysis for the cases mentioned above including the new results about
the topological charges. We finish this chapter and the paper with a conjecture about a commutative/non--commutative duality and remarks about approaching it.

\section{The Double Gyroid (DG) and Related Geometries and Material}
\subsection{The Geometry}
The gyroid is a triply periodic constant mean curvature surface that is embedded in  $\mathbb{R}^3$ \cite{Grosse}.
Figure \ref{fig1} shows a picture of the gyroid. It was discovered  in 1970 by Alan Schoen \cite{Schoen}. A single gyroid has symmetry group $I4_1{32}$  in Hermann-Maguin notation. Here the letter $I$ stands for bcc. The gyroid surface can be visualized by using the
level surface approximation \cite{lambert}
\begin{equation}
L_t:\sin x\cos y + \sin y \cos z+ \sin z \cos x=t
\label{level}
\end{equation}
In nature the single gyroid was observed as an interface for di--block co--polymers \cite {Hajduk}.
The {\bf double gyroid} consists of two mutually non--intersecting embedded gyroids. Its symmetry group is
$Ia\bar3d$ where the extra symmetry comes from interchanging
the two gyroids. It also has a level surface approximation which is given by the above expression (\ref{level}) with
$L_{w}$ and $L_{-w}$ for $0\leq w< \sqrt {2}$. The picture on the left hand side of Figure \ref{fig1} is actually a double gyroid or a ``thick"  surface.
\begin{figure}[t]
\includegraphics[height=2in]{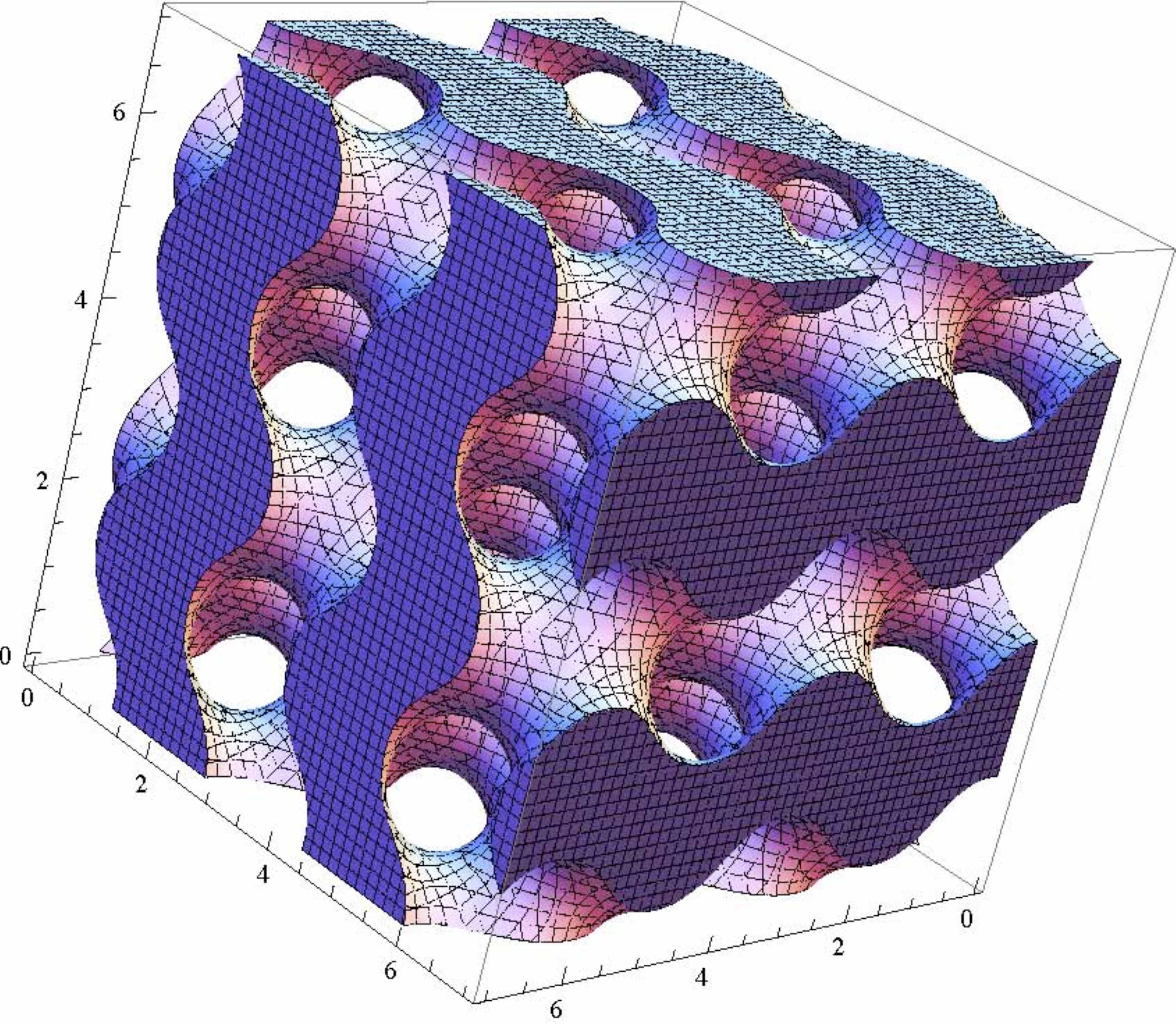}
\hspace{1cm}
\includegraphics[height=2in]{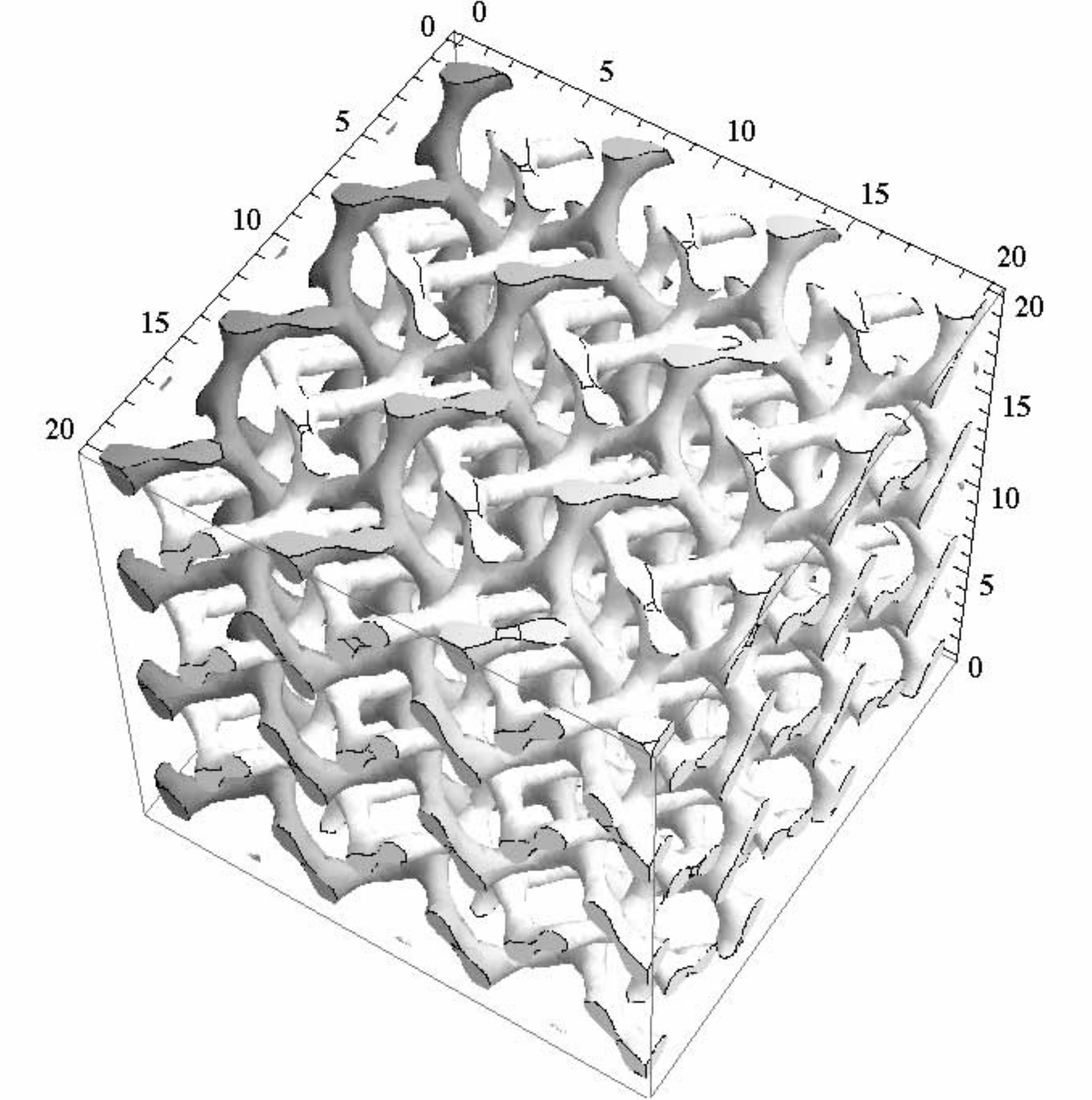}
\caption{The fat gyroid surface $W$ (left) and the two channel systems $C_+$ and $C_-$ (right)}
\label{fig1}
\end{figure}

Let us fix some notation.
We will denote by  $S=S_1 \amalg S_2$
 the double gyroid surface. Its complement $C=\mathbb{R}^3\setminus S$ has three connected components, which we will call $C_+,C_- $ and $W$. $W$
 can be thought of as a ``thickened'' (fat) surface which we will refer to as DG wall. There is a deformation retract of $W$ onto a single gyroid.

There are also two channel systems $C_+$ and
$C_-$, shown in Figure \ref{fig1}. These channels form Y-junctions where three channels meet under a $120$ degree angle.  Each of these channel systems
can be deformation retracted to a  skeletal graph $\Gamma_{\pm}$. We will concentrate on one of these channels and its skeletal graph $\Gamma_+$, shown in Figure \ref{fig2}.

 \begin{figure}[t]
 \includegraphics[height=2in]{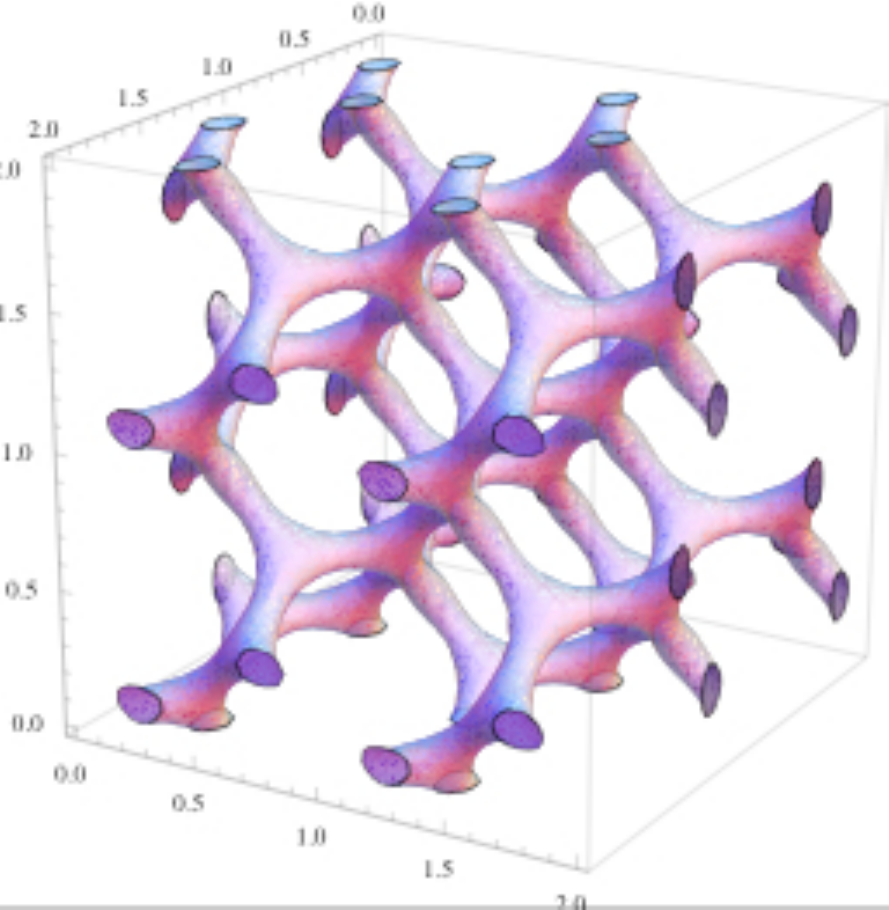}
 \hspace{2cm}
\includegraphics[width=2in]{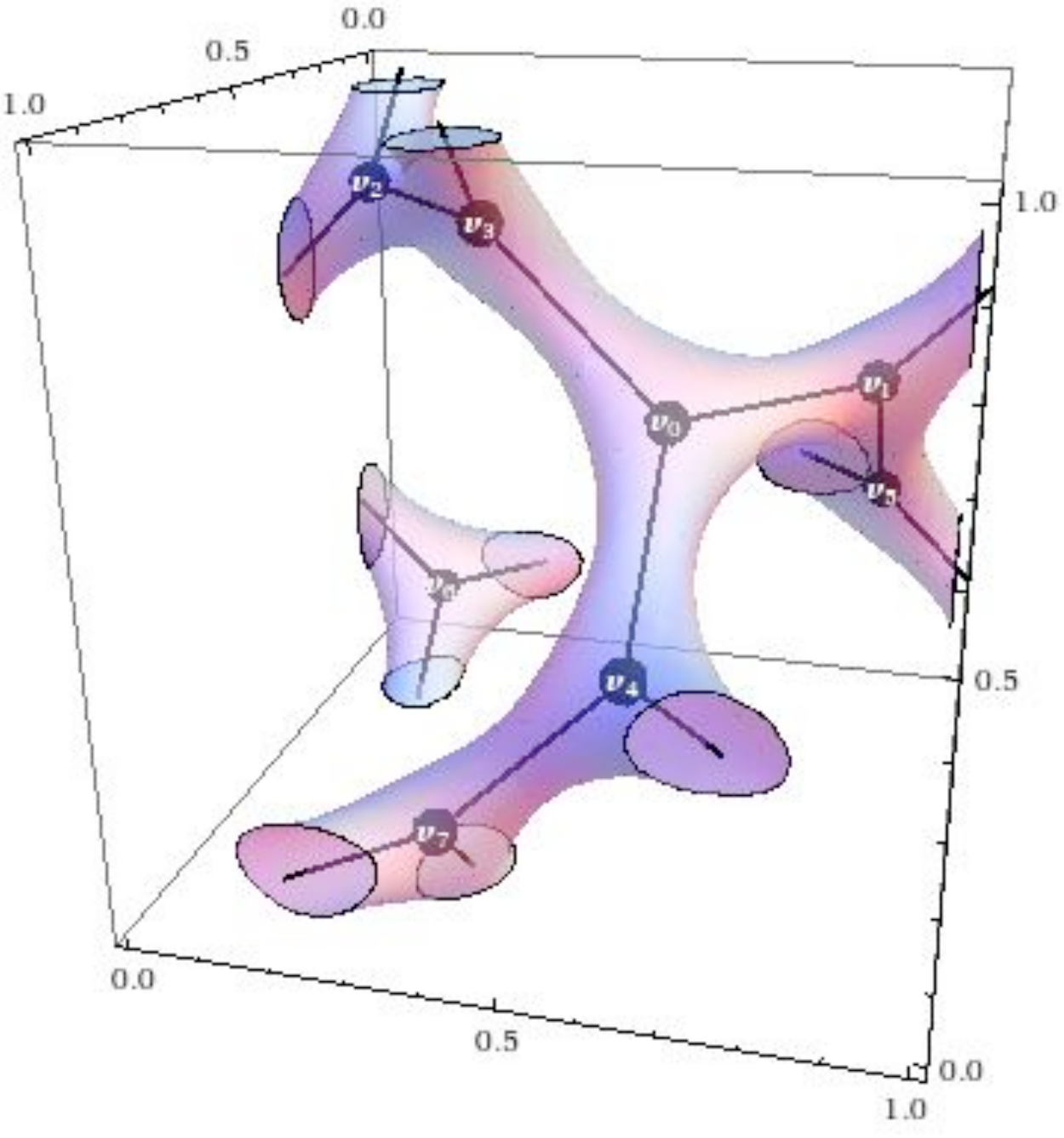}
\caption{One of the two channels (left) and its skeletal graph in the unit cell (right) }
\label{fig2}
\end{figure}

\subsection{The Material and Production}

A solid-state double gyroid can be synthesized by self-assembly at the nanoscale,
as demonstrated by Urade et al. \cite{Hillhouse}.
The first step is production of a nanoporous silica film with the structure of
unidirectionally cotracted double gyroid (DG) with lattice constant of about 18 nm.
The pores in the structure can then be filled with other materials to form nanowires.
Fabrication of platinum DG nanowires by electrodeposition
has been demonstrated in \cite{Hillhouse},
where it has also been mentioned that the process can be used for other metals
or semiconductors.

\subsection{Related Geometries: the P and D surfaces}

There are two other triply periodic self--dual and symmetric CMC surfaces- the cubic (P) and the diamond (D) network. They are shown in Figure \ref{fig3} together with their wire networks obtained in the same way as for the gyroid. Here we summarize the results from \cite{kkwk2}.

The P surface has a complement which has two connected components each of which can be retracted to the simple cubical graph whose vertices are the
integer lattice $\Z^3\subset \R^3$. The translational group is  $\Z^3$ in this embedding, so it reduces to the case of a Bravais lattice.

The D surface has a complement consisting of two channels each of which can be retracted to the diamond lattice $\Gamma_{\diamond}$.  The diamond lattice
is given by two copies of the fcc lattice,
where the second fcc is  the shift by $\frac{1}{4}(1,1,1)$ of the standard fcc lattice, see Figure \ref{fig3}. The edges
are nearest neighbor edges.
 The symmetry group is $Fd\bar3m$.

\begin{figure}[h]
\hspace{1cm}
\includegraphics[height=2in]{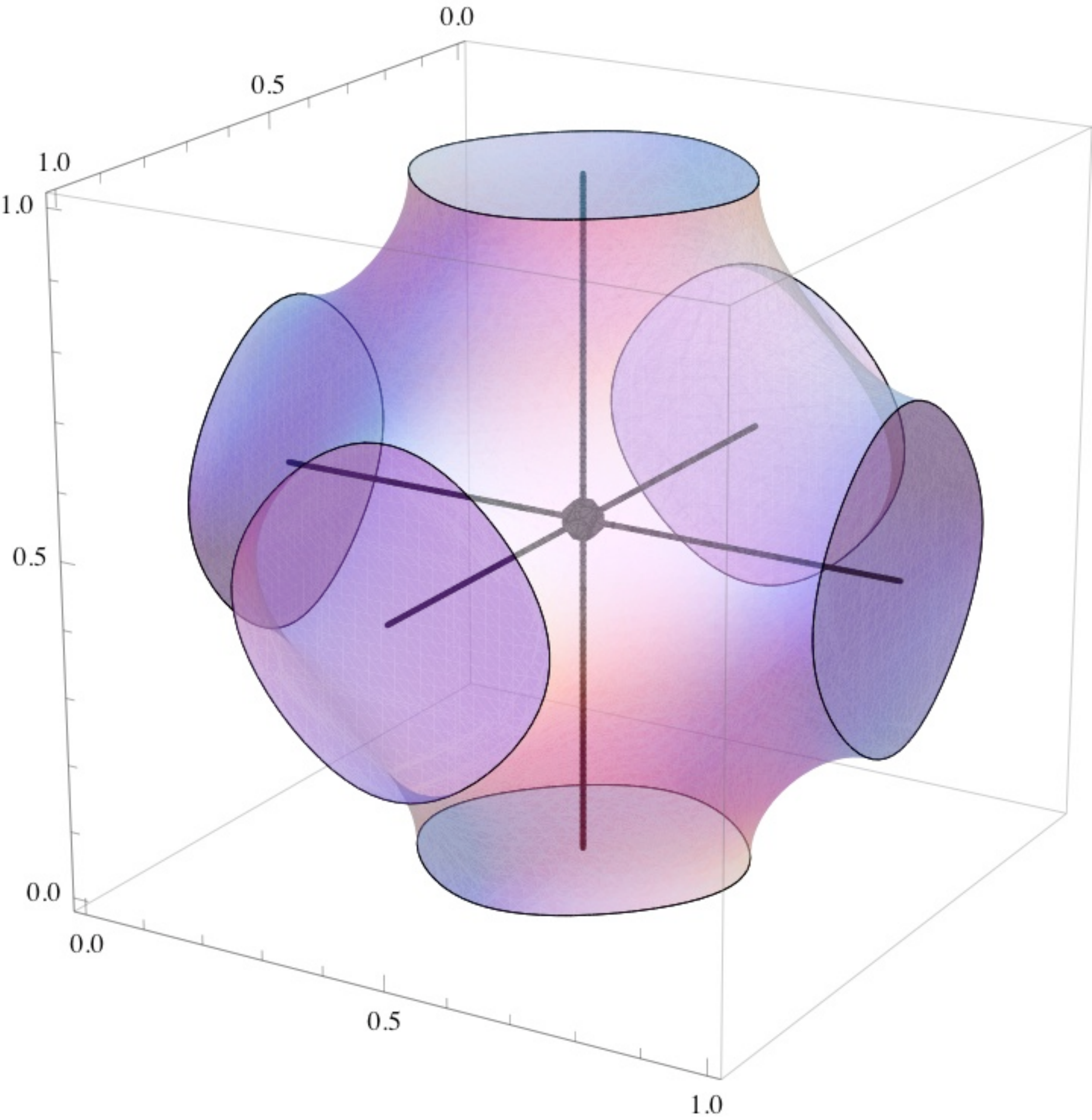}
\hspace{2cm}
\includegraphics[height=2in]{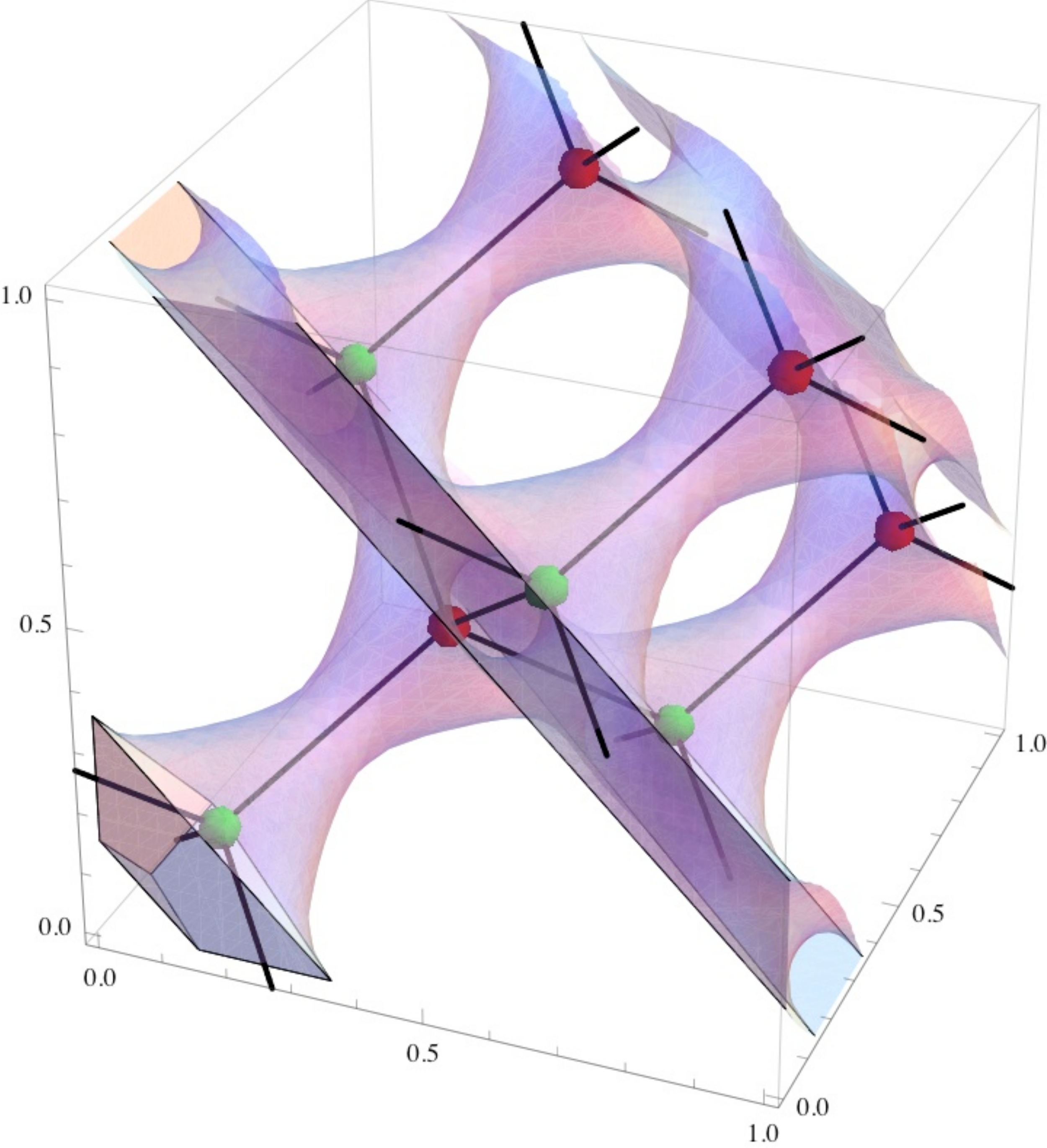}
\caption{The cubic (P)   (left) and the diamond (D) wire network (right)}
\label{fig3}
\end{figure}
\subsection{Graphene}
Graphene consists of one-atom thick planar sheets of carbon atoms that are densely packed in a honeycomb
crystal lattice. This two--dimensional material has attracted much interest recently, partially because of the existence of Dirac points where excitations show a linear dispersion relation.
Its electronic properties are described by a Harper Hamiltonian: see the review
\cite{CastroNeto} and references therein.
Here we will reproduce some of the known facts, such as the Dirac points using our
non--commutative geometry machine.

\section{Mathematical Model and Generalization: Graphs and Groupoid Representation}

\subsection{Discrete model  and Harper Hamiltonian}
We will now describe how to obtain the Harper Hamiltonian for any given graph $\Gamma\in \R^n$ with a given maximal translation group $L\simeq \Z^n$ \cite{Harper}. We will start with the commutative case without an external field, and then progress to  non--commutative case where the graph is placed in a constant external magnetic field. The mathematical set--up we will describe below can be understood in terms of Weyl quantization and Peierls substitution in physics \cite{PST}. Without the magnetic field the Harper Hamiltonian is given by translations, but in the presence of a magnetic field all translations turn into magnetic translations or Wannier operators, which cease to commute with each other.

Mathematically the discretization by the above process yields
 the Hilbert space $\H=\ell^2(V(\Gamma))$, where $V(\Gamma)$ are the vertices of $\Gamma$,
 and a projective representation of the translation group $L$ as well as an operator $H$, the Harper Hamiltonian. Concretely, the elements $l$ of $L$ act
on the functions $\Psi$ via the usual translations $T_l: T_l\psi(l')=\Psi(l-l')$.

\subsection{Quotient Graph and Harper Hamiltonian}
In general, given a embedded graph $\Gamma\in \R^n$, with a given maximal translation group $L\simeq \Z^n$, we consider the quotient graph $\bar\Gamma:=\Gamma/L$ and the projection $\pi:\G \to  \bG$.
 The quotient graphs for our four main examples are given in Figure \ref{quotientgraphs}.
 \begin{figure}[h]
\includegraphics[width=\textwidth]{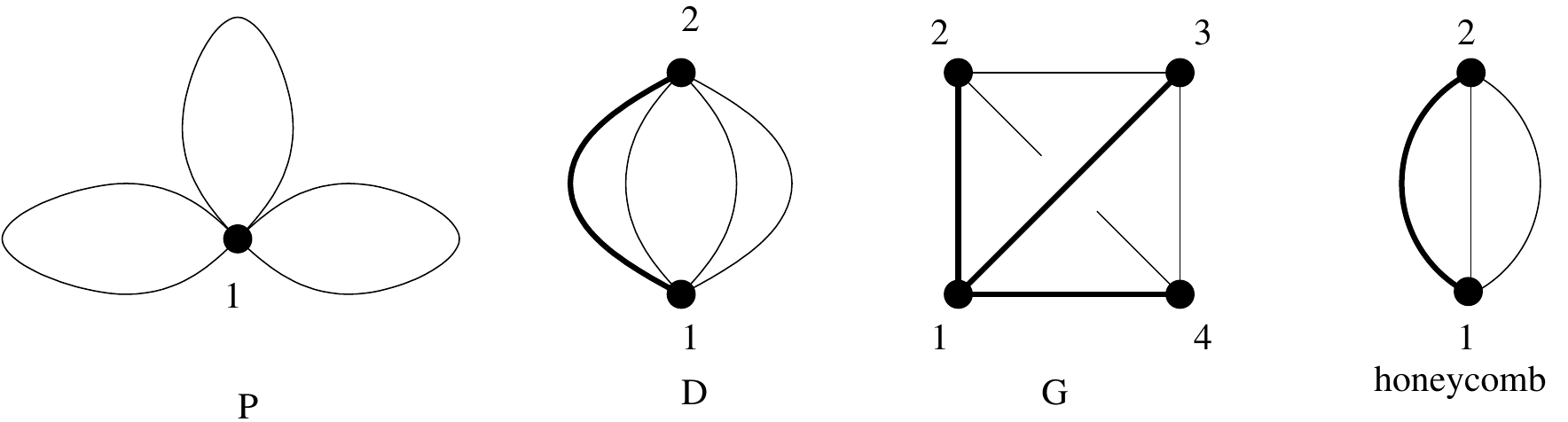}
 \caption{The quotient graphs of the  P,D,G surfaces and the honeycomb lattice, together with a spanning tree and an order of the vertices.} \label{quotientgraphs}
 \end{figure}

 The vertices of this graph
are in 1--1 correspondence with vertices or sites of $\Gamma$ in a fundamental cell.
We can think of the graph $\bar\Gamma$ as embedded into $T^n=\R^n/\Z^n$. Each edge $e$ of $\bar \Gamma$
lifts to a pair of edge vectors $\vec{e},\cev{e}=-\vec{e}$ where the underlying line segment is any lift
of $e$ to $\Gamma$. This is well defined since any two lifts differ by a translation.

To each vertex $v\in \bG$ we can associate the Hilbert space $\H_v:=\ell^2(\pi^{-1}(v))$.
Then the whole Hilbert space $\H$ decomposes as
\begin{equation}
\H=\bigoplus_{v\text{ vertex of } \bG} \H_v
\end{equation}
Since all the $\H_v$ are separable Hilbert spaces, they are all isomorphic.

The Harper Hamiltonian is then given as follows.
For each edge $e$ between two vertices $v$ and $w$
of $\bG$ let $T_{\vec{e}}$ be the translation operator from $\H_w \to \H_v$.
This extends to an operator $\hat T_{\vec{e}}$ on $\H$ via $\hat T_{\vec{e}}=i_{\bar v} T_{\vec{e}} P_{\bar w}$
where $i_{\bar v}:\H_{\bar v}\to \H$ is the inclusion and $P_{\bar w}:\H\to \H_{\bar w}$ is the projection.

The Harper Hamiltonian is
\begin{equation}
H=\sum_{e\in E} \hat T_{\vec {e}}+\hat T_{-\vec{e}}
\end{equation}

\subsection{Harper Hamiltonian in the presence of a magnetic field}
Adding a constant magnetic field requires a slightly different definition of the Harper Hamiltonian.  We will use projective translation operators whose commutators include the fluxes of the magnetic field as follows: We define  a 2--cocycle $\alpha_B\in Z^2(T,U(1))$ by a two--form $\Theta$. Such a two--form is given by a skew symmetric matrix $\hat\Theta$
with $\Theta=\hat \Theta_{ij} dx_i\wedge dx_j$. We let $B=2\pi \Theta$ where $B$ is the norm of the magnetic field. In this way we obtain a two--cocycle $\alpha_B\in Z^2(\R^n,U(1))$: $\alpha_B(u,v)=\exp(\frac{i}{2} B(u,v))$.

We define magnetic translations by starting from $A$, which is  a potential for
$B$ (on $\R^n$). The magnetic translation partial isometry is now acting on a wave function as
$$
U_{l'}\psi(l)=e^{-i \int_l^{(l-l')}A}\, \psi(l-l')
$$

The magnetic Harper operator is defined as
\begin{equation}
H=\sum_{e \text { edges of }\bG} U_{\vec{e}}+U_{\cev{e}}
\end{equation}

\subsection{Generalization: Groupoid and quiver representations}
In the setting above, which we call the geometric examples, we have  distilled the following data: a finite graph $\bG$, the translational groups $L$ and a projective representation of it on $\H=\bigoplus \H_v$ and finally the Hamiltonian $H$.

We will now explore the possibility of obtaining such data from a more general setup. There are two
ways to do this: in terms of groupoids or in terms of quivers.

\subsubsection{Groupoid representation} Recall that a groupoid is a category whose morphisms are
all invertible. A representation of a groupoid is a functor from this category into a linear category.
In our case this will be the category of separable Hilbert spaces which is the full subcategory of the category of vector spaces whose objects are separable Hilbert spaces.

A graph  $\bG$ (here $\bG$ need not be finite) determines a groupoid $\Gpd$ as follows. The objects are the vertices of $\G$.
The morphisms are {\em generated} by the edges. That is for each oriented edge between $v$ and $w$ there is one generator $\phi_{\vec{e}}$
 in $Hom(v,w)$. The morphisms in this category are then the composable words in the $\phi_{\vec{e}}$ where composable means that the source of a letter is the target of the predecessor, with the
 relations that
 \begin{equation}
 \label{inveq}
 \phi_{\vec{e}}\phi_{\cev{e}}=id_{v}\in Hom(v,v),  \text{ the identity element}
\end{equation}
What this means is that the morphisms are the paths on $\G$ up to homotopy, with the constant path yielding the identity.

A  groupoid representation of $\Gpd$
it in separable Hilbert spaces then assigns to each vertex $v$ of $\bG$ a separable Hilbert space $\H_v$ and to each oriented edge $\vec{e}$ from $v$ to $w$
a morphism $\Phi_{\vec{e}}\in Hom(\H_v,\H_w)$ with the relation that $\Phi_{\vec{e}}\Phi_{\cev{e}}=id_{\H_v}$.

The groupoid representation is unitary if all the $\Phi_{\vec{e}}$ are.

\begin{rmk}
Notice that there is  an involution $\ast$
on the morphisms, by transposing the word and reversing the orientation of each letter.
So we can only look at involutive functors, that is functors which send $\ast$ to $\dagger$,
which guarantees that the representation is unitary.
\end{rmk}

\subsubsection{Quiver representation} There is a way to formulate this in quiver language.
Given a graph $\bG$ and an arbitrary choice of directions for the edges determines a quiver.
Now one can construct the double of the quiver, where each oriented edge is doubled with reverse
orientation. If we started from a graph, this means that each unoriented edge $e$ is replaced by the two oriented edges $\vec{e}$ and $\cev{e}$. Now the double of the quiver is independent of the original choice of orientation. It also has an involution on $\ast$ the set of its edges which is given by reversing orientation. The quiver representations we are looking at are those where $\ast$ goes to
$\dagger$.

\subsubsection{Hamiltonian of the representation}
Just as above we define

$$
H:=\sum_{e\in E(\bG)}\rho(\vec{e})+\rho(\cev{e}):\H \to \H
$$

\subsubsection{Representation of $\pi_1(\bar \Gamma)$}
If we fix a vertex $v_0$ of $\Gamma$ the groupoid representation naturally gives a representation
of $\pi_1(\bar \Gamma)$ as follows. Fix a set of generators of $\pi_1(\bar\Gamma)=\FF_{1-\chi(\bar \Gamma)}$  is the free group in $1-\chi$ generators. Each such generator $g_i$ is a directed simple
loop on the graph which is given by a sequence of directed edges $\vec{e}_{1i},\dots, \vec{e}_{n_ii}$.
Then $\rho(g_i)=\rho(\vec{e}_{1i})\circ \dots\circ\rho(\vec{e}_{n_ii})$ gives a representation of $\pi_1(\bG,v_0)$ on $\H_{v_0}$.

\begin{df}
We will denote the algebra generated by $\rho(\pi_1)$ by $\mathscr T$.
  We say the $\rho$ is maximal if the generators of $\pi_1$ map to linearly independent operators
  and that $\rho$ is of torus type if $\mathscr T=\TTheta$.
  \end{df}

If $\rho$ is of torus type then $\rho$ is a projective representation of $H_1(\bG)$,
the Abelianization of $\pi_1$. These are of a special type, namely those whose co--cycle
is given by a constant $B$ field as discussed in \cite{kkwk}.

In the geometric situation of  Chapters 3.1--3.3,  maximality is equivalent to the fact that the translational symmetry group is maximal.

\subsubsection{Spanning trees}
\label{sptreesec}
If we pick a rooted spanning tree of $\bG$ then we get isomorphisms
$\phi_{0v}:\H_{v_0}\simeq \H_v$ by using $\rho$ and concatenation along the unique shortest path of oriented edges from $v_0$ to $v$ in the spanning tree. Let $\Phi=\bigoplus_v \phi_{v0}:\H_{v_0}^{|V|}\to \H$ then
this isomorphisms yields a representation $\tilde \rho$ on $\H_{v_0}^{|V|}$ via pullback.

Likewise $\phi_{v0}$ induces an isomorphism of $\pi_1(\bG,v)$ and $\pi_1(\bG,v_0)$.
Using this identification, we get an representation $\hat \rho$ of $\sT$ on $\H$
and via pull-back with $\Phi$ on $\H_{v_0}^{|V|}$.

A  rooted spanning tree $(\t,v_0)$ also gives rise to one more bijection. This is between a set of (symmetric) generators of $\pi_1$ and the edges {\em not} in the spanning tree. The bijection is as follows.

If $\vec{e}$ is a directed edge from $v$ to $w$ then there is a generator $g_{\vec{e}}$ which is
given by the following path of ordered edges: (1) the unique shortest path in $\tau$  from $v_0$ to $v$ (2) $\vec{e}$
and (3) the  unique shortest path in $\tau$  from $w$ to $v_0$. It is clear that $g_{\vec{e}}=g^{-1}_{\cev{e}}$. By contracting the spanning tree, we see that this is indeed a set of
symmetric but otherwise independent generators.

For convenience, we set $g_{\vec{e}}=1$ if $e\in \tau$.

\subsubsection{Generalized Bellissard-Harper algebra}
Given a groupoid representation in separable Hilbert spaces of a finite graph $\bG$ we
call the $C^*$ algebra generated by  the operators $H$ and $\sT$ via $\hat \rho$ on $\H$ the Bellissard--Harper algebra of the pair $(\bG,\rho)$ and denote it by $\B$.

This general set gives the generalization of one of the results of \cite{kkwk}.

\begin{thm}
Any choice of spanning tree together with an order on the vertices gives rise to a faithful matrix
representation of $\B$ in $M_{|V|}(\sT)$.
\end{thm}

\begin{proof}
This follows from the fact that under $\Phi$, $\rho(\vec{e})$ gets transformed to
the matrix entry $\rho(g_{\vec{e}})$ between the copies of $\H_{v_0}$ corresponding to $\H_{v}$ and
$\H_{w}$ under $\Phi$. Enumerating these vertices yields a matrix.
\end{proof}

In the following given a rooted spanning tree $\tau$ we will only choose orders $<$ such that the root
is the first element. The resulting matrix Hamiltonian will be denoted by $H_{\tau,<}$.

\section{$C^*$--geometry}

\subsection{Non--commutative case}
The non--commutative geometry of such a quiver representation in general and the one stemming from
the geometric situation in particular is that of $\B$.

Just like in \cite{kkwk,kkwk2} one can now ask the question whether or not $\B$ is isomorphic
to the full matrix algebra and hence Morita equivalent to $\sT$ itself. In the geometric case $\sT=\TTheta$ is generically simple and led to the expectation ---which we proved in \cite{kkwk}---
that generically $\B=M_{|V|}(\TTheta)$. This of course need not be the case in general.

It stands to reason that other more complicated physical phenomena could be described by
such algebras.

It actually turns out that in  the geometric examples not only is the algebra indeed the full matrix
algebra at  irrational parameter values, but that there are even only finitely many or a dimension--$1$ subset of rational matrix parameters $\Theta$, where $\B\subsetneq M_k(\TTheta)$.

\begin{thm} \cite{kkwk,kkwk2}
For the geometric cases of the G surface and the honeycomb lattice
the Bellissard--Harper algebra is the full matrix algebra except at finitely many values of $\Theta$
given in Chapter 6. For the P surface and all Bravais lattices $\B=\sT=M_1(\sT)$.
For the $D$ surface, the set of values of $\Theta$ for which    $\B\subsetneq M_2(\TTheta)$
is given by 6 one dimensional families and finitely many special points (also listed in Chapter 6). If $\B_{\Theta}$ is the full matrix algebra then it is Morita equivalent to $\TTheta$.
\end{thm}

\begin{rmk}
Note that except for the P and general Bravais case, these families above give examples of continuos
variations of algebras whose $K$--theory does not vary continuously.
\end{rmk}

\subsubsection{K--theory labeling}
One application of the non--commutative approach is gap labeling by $K$--theory. If
the Hamiltonian $H$ has spectrum bounded from below, then each gap in the spectrum
gives rise to a projector $P_{<E}$ onto the Eigenspaces with Eigenvalues less than any fixed value $E$ in the gap, see e.g. \cite{B}. The gap labeling then associates the $K$--theory class of $P_{<E}$
to the gap.

By the above result, via the inclusion $\B\hookrightarrow M_k(\sT)$ the projector $P_{<E}$
also gives rise to a $K$--theory class in $K(M_k(\sT))\simeq K(\sT)$. Using this embedding, one
can deduce analogues of the famous Hofstadter Butterfly.

\begin{thm}
If $(\bG,\rho)$ is toric non--degenerate, then the Hamiltonian $H$ as an operator on $\H$
has only finitely many gaps if the magnetic field is rational in the sense that the matrix $\Theta$
is rational.
\end{thm}

\subsection{Commutative case}
\label{commutative}
If $\B$ is commutative, for instance if $\Theta=0$ in the geometric situation, then
by the Gel'fand--Naimark theorem, there is a compact\footnote{$\B$ is unital} Hausdorff space $X$,
such that $\B\simeq C^*(X)$. The points of $X$ can be thought of as characters, i.e.\ $C^*$--homomorphisms $\chi:\B\to \C$. More precisely these characters are in bijection with the maximal
ideals of $\B$ which are the points. If we wish to make this distinction, we write  $p_{\chi}$ for the point of $X$ corresponding to the character $\chi$ and vice--versa $\chi_p$ for the character corresponding to $p$.

Likewise there is a space $T$ which corresponds to the $C^*$--algebra $\sT$. In the geometric case
$T=T^n=\R^n/L$.

As usual the correspondence between the algebra of functions and the spaces is contravariant.
This means that the inclusion $\hat\rho:\sT\to \B$ gives rise to a morphism $X\to T$.
If $(\rho,\Gamma)$ is maximal, then $\sT\to \B$ is injective and hence $X\to T$ is surjective.

Furthermore let us consider the algebra $\sT^n$ given by the direct sum of $n$ copies of $\sT$.
The space corresponding to this algebra is simply $T\amalg \dots \amalg T$ $n$--times.

Since after choosing an order and a rooted spanning tree $\B\subset M_{|V|}(\sT)$,
we can lift any character $\chi$ of $\sT$ to a $C^*$-homomorphism:
$\hat \chi:M_{|V|}(\sT)\to M_{|V|}(\C)$ of $\B$ by applying $\chi$ to each entry.

\begin{df}
We call a point $\chi$ of $T$ degenerate if $\hat\chi(H)$ has less than  $|V|$ distinct Eigenvalues.
\end{df}

Repeating the proof of \cite{kkwk} we arrive at the following

\begin{thm}
If $(\rho,\bG)$ is maximal the map $\pi:X\to T$ is ramified over the degenerate points and
furthermore $X$ is the quotient of the trivial cover $T^n$ where the identifications are made in
the fibers over degenerate points and moreover these correspond to the degeneracies of $H$ over these points.

In other words $X$ can be thought of as the spectrum of the family of Hamiltonians $H(p)=
\chi_p(H)$
parameterized over $T$.
\end{thm}

The key ingredient is the image of $H$ under the map $\B\to \sT^n$

\begin{equation}
H\mapsto \sum_i \lambda_i e_i
\end{equation}
where $e_i$ are the idempotents corresponding to the  i--th component and $\lambda_i$  is the $i$--th Eigenvalue.

\subsubsection{Bundles, K-theoretic and Cohomology Valued Charges}

Let $T_{\rm deg}$ be the subset of degenerate points on $T$ and let $X_{\rm deg}=\pi^{-1}(T_{\rm deg})$
be the closed possibly singular locus of $X$. Then the restriction $\pi:X_0:=X\setminus X_{\rm deg}\to T_0:=T\setminus T_{\rm deg}$ is the trivial $k$--fold cover. On this restriction the $C(X\setminus X_{\rm deg})$ contains pairwise orthogonal projections $P_i$ such that $H=\sum_i P_i$. Each of these $P_i$ defines a rank $1$ sub--bundle $L_i$ of the trivial bundle $X_0\times \C$ which is
the Eigenbundle  corresponding to the Eigenvalue $\lambda_i$. The projector or equivalently the bundle $L_i$ defines
an element in $K$--theory $[L_i]\in K(T_0)$. We will continue with the geometric interpretation of line bundles and K--theory here, although in forthcoming analysis we will concentrate on the $C^*$ version of $K$--theory in oder to move to a non--commutative setup.

We call the classes $[L_i]$ the K--theoretic charges and the associated Chern classes
$\beta_i:=c_1(L_i)\in H^2(T_{0})$ the cohomological charges. We also let $C=\bigoplus_i L_i$, and
 $[C]\in K(T_0)$ be its class in $K$--theory. Finally we define the polynomial invariant
 $Q_c(t_i)=\prod_i(1+t_i\beta_i)\in H^{\rm ev}(T_0)[t_i]$.
This class contains all the cohomological information of the $L_i$ and $C$.

\begin{rmk}
If $H(t)$ is not Hermitian, we also need the condition $\pi_1(T_0)=0$  in order for the
characteristic polynomial to be irreducible over $C(T_0)$ which is necessary to define the $P_i$.
\end{rmk}

\begin{rmk}
We assumed that the Hamiltonians are generically non--degenerate. It is sufficient to assume that the ranks of the Eigenbundles are generically constant. In this case, we have vector bundles $V_i$ and total Chern classes $c(V_i)$.
\end{rmk}

\subsection{Berry connection, topological charge and slicing}

One can try to get numerical information about $Q_c$ and the $\beta_i$ by pairing them with
appropriate homology classes. For this it is easier to assume that we are dealing with oriented manifolds. If we furthermore have a differentiable structure, we know that we can evaluate Chern classes
by using Chern--Weil theory.

Of course the charges are trivial if $T_0$ has vanishing second cohomology (e.g. if $T_0$ is 2--connected).
In that case the Chern classes $\beta_i$ vanish and the line bundles $[L_i]$ are trivializable.
This is the case in some examples, notably the honeycomb. The effect is that the line bundles
are trivializable and the associated points of degeneracy are not topologically stable, see \S\ref{stablepar}.

The two--torus or the two--sphere do however have non--vanishing $H^2$ and thus are prime candidates
to detect first Chern classes.

Furthermore if there is a differentiable structure, applying  Chern--Weil Theory to the particular case of a line bundle,  we can evaluate the first Chern class of a line bundle with a connection on a 2--dimensional submanifold by pulling back, i.e.\ restricting, the line bundle to the surface and integrating the curvature form of the connection.

\subsubsection{Berry connection}

Following Berry \cite{Berry} we can use the connection provided by adiabatic transport.
It was Berry's insight that this connection is indeed not always trivial
 and produces the so--called Berry phase as a possible monodromy. In the reinterpretation
of Simon \cite{simon} this connection computes exactly the first Chern class of the line bundle $L_i$,
which is the only obstruction for $L_i$ and hence the monodromy to be trivial.

\subsection{Topological charges}
There are several ways to get a scalar charge which one can exploit. Since the Chern classes have
even degree, they will always produce zero when paired with odd dimensional homology classes.
Thus we have to ensure that we use even dimensional cycles to integrate over. (Here integration means pairing with the fundamental class).

Assume $T$ is compact orientable potentially with boundary and that $T_{\rm deg}$ is in codimension
at least $1$; i.e.\ $T$ is generically non--degenerate.
We furthermore assume that $\Tdeg\cap \del T=\emptyset$. Then $T_0$ is an orientable
manifold with boundary.
Let $N$ be a tubular neighborhood of $T_{\rm deg}$ in $T$.
Then $B=T\setminus N$ is a compact sub--manifold with boundary
$\del B=\del T\amalg \del \bar N$ where $\del \bar N=\bar N\setminus N$.
If $B$ has 2nd cohomology, we can pair the $\beta_i$ and $Q_c$ with suitable homology classes.

Assume that $T_{\rm deg}$ is a manifold with singularities.
If the smooth part of $T_{\rm deg}$ is of codimension $r$ then $\del \bar N$ is an $S^{r-1}$
bundle over the smooth part of $T_{\rm deg}$.

\subsubsection{Even dimensional $T$}

If $B$ is even dimensional, we can integrate over  $B$ itself  and consider $Q_B=\int_B Q_c(t_i)$,
the full $B$ charge.

If in particular $B$  is two--dimensional, we obtain all the individual charges $Q_i:=\int_B\beta_i$.

Following Simon \cite{simon} this if for instance the case for the quantum Hall effect.
Here $T=T^2$ has no degenerate locus and we have that $B=T$ can carry non--trivial line bundles.
Indeed the arguments of TKNN \cite{TKNN} establish the non--triviality of the corresponding line bundle.

\subsubsection{Odd dimensional $T$ with boundary}
If $T_0$ is odd dimensional, we can restrict the $L_i$ to
the boundary of $\del T$. Then the {\em boundary charge} is $\int_{\del T} Q_c(t)|_{\del T_0}$.

In the differentiable case, we represent $Q_c$ by a closed form $\omega=d\phi$; strictly speaking this is a polynomial form. Then since $B$ is odd dimensional, we have by Stokes' Theorem that
$0=\int_B \omega=\int_{\del B}\phi=\int_{\del T} \phi +\int_{\del \bar N}\phi$.
\begin{equation}
\label{localglobaleq}
\int_{\del T} \phi=-\int_{\del \bar N}\phi=\int_{-\del \bar N}\phi
\end{equation}
where $-{\del \bar N}$ has the outward orientation  viewed from $\bar N$.
Else we just use the usual pairing between the corresponding homologies and cohomologies.

If the boundary is empty, then we have that $\int_{\del \bar N}\phi=0$.

\subsubsection{Codimension 3 and local charges}
If the smooth part $T^{\rm sm}_{\rm deg}$ of $T_{\rm deg}$ is of codimension $3$ then we can restrict the $L_i$ to the fiber
$S^2=S^2(p)$ over any point $p$ of $T^{\rm sm}_{\rm deg}$. We call $\int_{S^2(p)}L_i|_{S^2}(p)$
the $i$--th local charge at $p$ and $\int_{S^2(p)}Q_c|_{S^2}(p)$ the total local charge.

\subsubsection{Isolated critical points in dimension $3$}
For {\em isolated} critical points of $T_{\rm deg}$  the local charges are just given by integrating over small spheres around these points.
If $T_{\rm deg}$ consists only of isolated critical points, then formula (\ref{localglobaleq})
states that the boundary charge is the sum over the local charges.
If moreover the boundary is empty, this means that the sum of all the local charges is $0$.
This is the case for the gyroid.

\subsection{Probing with surfaces}
In order to detect the $K$--theoretic charges, we can send them to cohomology using the Chern classes and then detect them by using embedded surfaces.
Explicitly, if $\Sigma$ is an oriented compact surface and $i:\Sigma\to T$ is an embedding, then

\begin{equation}
\label{surfacechargeeq}
Q_{\Sigma,i}:=\int_{\Sigma} i^*c_1(L_i)=\langle c_1(L_i), i_*([\Sigma])\rangle
\end{equation}
where $\langle \;,\;\rangle$ is the standard pairing between cohomology and homology.
Notice that by the results of Thom \cite{Thom} all second  homology classes are of this type even over $\Z$. In general, one has to take at least $\Q$ coefficients to ensure that all these integrals determine the cohomology class uniquely.

\subsubsection{Slicing}
A slicing  for $T$ is a smooth codimension $1$ foliation by compact oriented manifolds of $T$ which
has a global transverse section $S$ and the leaves of the foliation generically do not intersect $T_{\rm deg}$. For this we need the Euler characteristic to be $0$, which is in particular the case
for all odd dimensional compact manifolds.

For $s\in S$ let $T_s$ be the leaf of $s$ and $i_s$ be the inclusion,
we can consider the pullback of $C$ and consider
\begin{equation}
Q_s:=\int_{T_s} i^*C
\end{equation}
which is the total Chern class of the slice.
An interesting situation arises if
\begin{enumerate}
\item $T_s$ generically does not intersect $T_{\rm deg}$
\item Any component of $T_{\rm deg}$ is contained between some pair of slices. That is for a component $T'\subset T_{\rm deg}$ there are $s_1,s_2$ and an n--dimensional submanifold $M_{T}'$
    of $T$ with boundaries, such that $M_{T'}\cap T_{\rm deg}=T_0$, and $\del M_{T'}\cap T_{\rm deg}=\emptyset$ and $\del M_{T'}=T_{s_1}-T_{s_2}$.
\end{enumerate}

In this case, by using Stokes' Theorem we get that the total contribution of $T'$

\begin{equation}
\int_{N\cap M_{\t'}}Q_s=Q_{s_1}-Q_{s_2}
\end{equation}

Now (\ref{sliceeq}) is a great tool to numerically find $T_{\rm deg}$.
For this one just runs through the $s\in S$ and looks for jumps in $Q_s$.

\subsubsection{$T_{deg}$ of  codimension $3$}

If $T'$ is smooth
then the total charge is
\begin{equation}
\label{sliceeq}
\int_{Tdeg}(\int_{S^2(p)}C|_{S^2}(p))dp=Q_{s_1}-Q_{s_2}
\end{equation}

If we are in dimension $3$ then codimension $3$ means that the degenerate locus consists of only
isolated critical points. Here the equation (\ref{sliceeq}) simplifies to just a finite sum over the critical points.

If furthermore the critical points are $A_1$ singularities, see \S\ref{swallowpar}, then the jumps in the charge
are from $\pm1$ to $\mp 1$, as calculated in \cite{simon,grove} depending on if one calculates for the upper or lower band and the chosen orientation/parameterization.

\subsubsection{3--dimensional torus models}
If we have that $T=T^3$ the situation is especially nice. It is fibred by $T^2$s in  any sprojections $S^1\times S^1\times S^1\to S^1$.
The inclusion of fibers, say in the three coordinate projections, actually generates the whole cohomology of $T^3$ which has non--vanishing 2nd cohomology $H^2(T^3)\simeq \Z^3$.
 In contrast to the two--torus where puncturing kills the 2nd cohomology a punctured three torus actually still has second cohomology. It is given explicitly in the proof of the theorem below. This is a main difference between graphene and the gyroid,  see below. One has to  be sure however, that the condition of generically not intersecting the degenerate locus is not violated. This is for instance the case for the $D$--surface, see below.

\begin{thm}
For a smooth variation with base $T^3$ with and only  finitely many degenerate points, the slicing method corresponding to a generic projection completely determines the $K$--theoretic charges and hence the line bundles $L_i$ up to isomorphism.
\end{thm}

\begin{proof}
If there are $m$ degenerate points $p_i$ then pick a generic projection and let $z_1,\dots, z_m\in S^1$ be the images of the $p_i$. Let $t_1,\dots, t_m$ be points in between the $z_i$, that is one point per component of $S^1\setminus \{p_i\}$. Consider the CW model of the torus, which has one 2--cell at height $t_i$ and
3--cells in between and $0$ and $1$ cells accordingly. Then $T_0=T\setminus\{p_i\}$ deformation retracts onto the 2--skeleton of this complex.
And the homology $H_2(T_0)$ is generated by exactly the $m$ two cells. Now the slicing method will give the paring with these two cells and
as the Poincar\'e paring is non--degenerate, we  the cohomology class of $c_1(L_i)$ is determined by these numbers and hence the line bundle up to isomorphism.
\end{proof}
Notice that the slicing only gives a finite set of numbers for each Eigenvalue, since the integral over the Chern--class is constant
in the components $S^1\setminus\{p_i\}$.

\subsection{Topological Stability}
\label{stablepar}
Having non--vanishing topological charges produces topological stability. If we perturb the Hamiltonian slightly by adding a small perturbation term $\lambda H_1$ and continuously vary $\lambda$ starting at $0$, then  $T_0$ does not move much ---for instance as a submanifold  of $T\times R$, see \S\ref{swallowpar}. In particular, there will be no new singular points in $T_0$
for small perturbation. The Eigenbundles over $T_0$ also vary continuously and hence so do their Chern classes. Since these are defined over $\Z$ they are actually locally constant, so that all the non--vanishing charges, scalar, K-theoretic or cohomological, must be preserved.

\section{Swallowtails and symmetries}

\subsection{Characteristic map  and Swallowtails}
\label{swallowpar}
In the commutative case, the locus $X_{\rm deg}$ has a nice characterization in terms of singularity theory, \cite{kkwk3}.

The key ingredient is embedding of $X$ into $T\times \R$ and the characteristic map.
Let $P(z,t)=det(zId-H(t))=z^k+b_{k-1}(t)z^{k-1}+\dots +b_0(t)$, let
$P(z-\frac{b_{k-1}}{k},z)=z^k+a_{k-2}(t)z^{k-2}+\dots +a_0(t)$ and let $g$ be the isomorphism
on $T\times \R$ which sends $(t,z)$ to $(t,z-\frac{b_{k-1}}{k})$.
The coefficients $a_{k-2}(t),\dots, a_{0}(t)$ define a map $\Xi:T\to \C^{k-1}$ called
the characteristic map. Identifying  $\C^{k-1}$ with the base of the miniversal unfolding
of the $A_{k-1}$ singularity, we obtain the following generalization of \cite{kkwk3}:

\begin{thm}
 \label{mainthm}
 The branched
cover $X\to B$ is equivalent via $g$ to the pull back of
 the miniversal unfolding of the $A_{k-1}$
 singularity along the characteristic map $\Xi$.

Moreover is the family of Hamiltonians is traceless, which is for example
the case if $\bG$ has no small loops ---that is edges which are a loop at one vertex---, the cover
is the pull--back on the nose.

Furthermore, if the graph is also simply laced, then $a_{k-1}=|E(\bG)|$ and the image of $T$ is contained
in that slice.
\end{thm}

This means that if $\Sigma\subset \C^{k-1}$ is the discriminant locus or swallowtail,
then $T_{\deg}=g^{-1}(\Xi^{-1}(\Sigma))$ and the fiber of $\pi$ over a point $t$ is
exactly $g^{-1}\pi_A^{-1}(\Xi(t))$ where $\pi_A$ is the projection of the miniversal unfolding.
In other words the fibers over degenerate points are identified with the corresponding fibers
over their image points in the swallowtail.

Using Grothendieck's characterization \cite{grothendieck} of the swallowtail as stratified  by lower order
singularities obtained by deleting edges in the corresponding Dynkin diagram, we obtain:

\begin{cor}
The only possible types of singularities for $(\bG,\rho)$ with traceless Hamiltonians in the spectrum are $(A_{r_1}, \dots, A_{r_s})$ with $\sum r_i\leq k-s$.
\end{cor}

\begin{rmk}
Theorem \ref{mainthm} and the corollary above can be viewed as a more precise statement of what is commonly referred to as the von Neumann--Wigner theorem. Namely the expectation that the degenerate locus is of codimension $3$.  This is the case for the full family of Hermitian Hamiltonians as shown in \cite{vNW}. In general the exact codimension depends on the whole family $T$ and is given precisely
as the preimage of $\Xi$. To be more precise locally it is the dimension of the intersection of the image under $\Xi$ with the swallowtail and the dimension of the fiber.
\end{rmk}

\begin{prop}
In the maximal toric case increasing the number of links to arbitrarily high values, the codimension of the degenerate locus $T_{deg}$ generically becomes $-\chi(\Gamma)$, so that the stable expected codimension of the critical locus is $1$.
\end{prop}

\begin{proof}
Since the domain of $\Xi$ is compact, so is the image. Its size is limited by the coefficients of
the Hamiltonian. The value of  $i,j$--th entry under $\hat\chi$ is sharply bounded by $l$ where $l$ is the number of edges between $v_i$ and $v_j$. As the number of edges grows this bound increases. This implies that the sharp bound on the coefficients $a_i$ also increases.
 If this is large enough, the image of $\Xi$ will fill out a bounded region of the complement of the swallowtail $\Sigma$ over which the discriminant is positive. Then the boundary of the image given by a part of the swallowtail $\Sigma$ will be of codimension $1$ and of dimension $|V_{\Gamma}|-2$.
The generic dimension of the fiber will be $dim(T)-(|V_{\Gamma}|-1)$. In total this gives the dimension of the critical locus as
$1-\chi(\Gamma)-|V_{\Gamma}|+1+|V_{\Gamma}|+2=-\chi(\Gamma)$.
\end{proof}

The test case of the triangular graph has been calculated in \cite{kkwk2} which gives an example of
the phenomenon described above.

\subsection{Characterizing Dirac points}
Physically very interesting singularities of $X$ are conical singularities, which are also called Dirac points. In order to find these singularities, we considered the ambient space $T\times \R$ and the function $P:T\times \R\to \R$.
 As we argued in \cite{kkwk3}, Dirac points in the spectrum are isolated Morse singularities of $P$ with signature $(+,-,\dots,-)$. That argument did not need the specifics of the geometric situation and
 hence generalizes.

 Notice that a necessary condition from the above is that there is an $A_1$ singularity in the fiber.
In addition one needs to check the signature.

\subsection{Symmetries and the re--gauging groupoid \cite{kkwk4}}

\subsubsection{General setup}
Going back to the embedding of $\BTheta$ into $M_k(\TTheta^n)$ the relevant matrix representation depended on the choice of a rooted spanning tree $(\tau, v_0)$ and an order $<$ on the vertices. We will now fix that the first element in that order is given by the root. In \cite{kkwk4} we showed that the re--gauging from $(\tau, <)$ to $(\tau',<')$ is given by conjugation by a unitary matrix $U_{\tau,<}^{\tau',<}$. These matrices are more complicated than just the permutation group and incorporate local gaugings. These are given by  diagonal matrices with invertible
elements in $\TTheta$ indexed by the vertices of the graph.

Moreover in this way, the automorphism group of $\Gamma$ acts by re--gaugings. Namely, if $\phi\in Aut(\Gamma)$ then given $(\tau,<)$, the image of $\tau$, $\phi(\tau)$,  and the push forward of the order, $\phi_*(>)$,  give rise a re--gauging by $U_{\tau,<}^{\phi(\tau),\phi_*(<)}$. Usually this action on a given Hamiltonian is not trivial, due to the fact that $\rho$ need not be trivial.

All these observations directly generalize to the more general case of a
groupoid representation $(\bar \Gamma,\rho)$. In this case $\BTheta$ is replaced by $\sT$.
The arguments of \cite{kkwk4} are not sensitive to the particular structure of $\TTheta$
and hence carry over to the more general situation. We summarize the logical steps here.
\subsubsection{Re--gauging groupoid}
The re--gaugings form a secondary groupoid, the re--gauging groupoid. Its objects are given by  tuples $(\tau,<)$ and between any two objects there is a unique morphism
$((\tau,<),(\tau',<'))$.
There is a morphism $\lambda$ to matrices with coefficients in $\sT$ by sending   $((\tau,<),(\tau',<'))$ to $U_{\tau,<}^{\tau',<}$. This morphism need not be a representation however, since we are
only guaranteed that $\lambda(g_1)\lambda(g_2)\lambda(g_1g_2)^{-1}$ is
 non--commutative 2--cocycle with values in $U(\sT)$, the unitary elements of $\sT$.
  The reason for this is that under the identification given in \S \ref{sptreesec} the re--gauging basically corresponds to an isomorphism of  $\pi_1(\bG,v_0)$ with $\pi_1(\bG,v'_0)$  along a path, $v_0$ and $v_0'$ being the roots of $\tau$ and $\tau'$ respectively. Concatenating the isomorphisms along these paths as above, we end up with an isomorphism under a loop; but this is precisely conjugation with an element of $\pi_1(\bG,v_0)$.
  In the representation, this element becomes an element in $U(\sT)$.

\subsubsection{Projective Groupoid Representations}
In the commutative case the cocycle above gives rise to a central extension by
$U(\sT)$ and the matrices $U_{\tau,<}^{\tau',<'}$ give a representation in $M_k(\sT)$ of the central extension.

Evaluating with a character $\hat \chi$, the extension becomes an extension by $U(1)$ and
the matrices $\hat\chi(U_{\tau,<}^{\tau',<})$ form a projective representation of the groupoid
in $M_k(\C)$.

 \subsubsection{Stabilizer Groups, Lifts, Projective Actions and Group Extensions}

If we have a fixed point, that is a Hamiltonian that is invariant under the action of  non--trivial groupoid elements, then these form a {\em group} of re--gaugings. Technically the representation of  stabilizer subgroupoid factors through the group given by identification of all objects in that groupoid to one point.

In order to find such a stabilizer group, we look for an automorphism of $T$ which compensates the
re--gauging by automorphisms of $\bG$. That is given an automorphism $\phi$ of $\bG$ let $\Phi_{\tau',<'}^{\tau,<}$ be the associated re--gauging. We then look for an automorphism
$\Psi_{\tau',<'}^{\tau,<}$ of $T$ such that
\begin{equation}
\hat\chi_t(\Phi_{\tau',<'}^{\tau,<}(H_{\tau,<}))=\hat\chi_{\Psi_{\tau',<'}^{\tau,<}(t)}(H_{\tau,<})
\end{equation}

This is done for one orbit of $(\tau,<)$ under $Aut(\bG)$. This tool is most effective is the graphs are completely symmetric, like the cases we considered.

If we find such a lift of the automorphism group $Aut(\bG)\to Aut(T)$, then we can look for points
of enhanced symmetry. If $t\in T$ has a non--trivial stabilizer group under this action of $Aut(\bG)$
then the matrix $\hat\chi_t(H_{\tau,<})$ has a non--trivial re--gauging fixed group.
This action by conjugation yields a projective representation of the stabilizer group.

Given such a projective representation, we know that it is a representation of a central $U(1)$
extension of the stabilizer group. If the stabilizer group is finite, we would furthermore like to find
a smaller if possible finite group which already carries the representation. That is an extension of the stabilizer group by a finite group. For this one uses the theory of Schur multipliers.

The upshot is that the isotypical decomposition of the representation has to be commensurate with the Eigenspace decomposition of the Hamiltonian -- for that particular value $t\in T$. Practically this means that on one hand if in the given representation there are irreps of dimension bigger than one, one can infer that there are degeneracies in the spectrum of at least these dimensions.
On the other hand, the one dimensional isotypical components  fix Eigenvectors and hence make it easy to find the Eigenvalues. In general of course one only has to diagonalize the Hamiltonian inside the isotypical summands.

In the geometric examples, we showed in \cite{kkwk4} that all the degeneracies can be explained as being forced by these enhanced re--gauging symmetries.

\section{Results and the conjectured NC/C Duality}

Let us summarize our results for the different quantum wire networks, honeycomb, P, D and G. The basis are the results from \cite{kkwk,kkwk3,kkwk4,kkwk2} and the new analysis for the topological charges.

\subsection{The Honeycomb Lattice}
\subsubsection{The commutative case}
In this case the space $X$ is a double cover of the torus $T^2$ ramified at two points.
These two points are $A_1$ singularities and Dirac points.

$T_0$ is $T^2$ with two points removed, so $H^2(T_0)=0$ and so the all charges vanish and the two Dirac points are in general not topologically stable.

There has been an investigation of
deformation directions which do not destroy these points \cite{Fefferman}.
In our setup this means the following:
the
characteristic map has its image in $[-9,0]$  where the swallowtail for $A_1$ is the point $0$.
One only considers deformations which still have $0$ in the image of the characteristic map.

At the Dirac points there is an enhanced symmetry which is Abelian, so it does not have any higher dimensional irreps, but the isotypical decomposition is fully decomposed and forces the double degeneracy at the Dirac points due form of the Hamiltonian.

\subsubsection{Noncommutative case}
Generically $\BTheta=\TTheta^2$. In order to give the degenerate points, let $-e_1:=(1,0),e_2=\frac{1}{2}(1,\sqrt{3}),e_3:=\frac{1}{2}(1,-\sqrt{3})$ be the lattice vectors and $f_2:=e_2-e_1=\frac{1}{2}(-3,\sqrt{3})$,
$f_3:=e_3-e_1=\frac{1}{2}(3,\sqrt{3})$ the period vectors of the honeycomb. The parameters we need are
\begin{equation}
\label{notationeq}
\theta:=\hat\Theta(f_2,f_3), \quad q:=e^{2\pi i \theta} \quad\text{\it and} \quad
\phi=\hat\Theta(-e_1,e_2), \quad \chi:=e^{i\pi \phi},  \text{\it thus } \quad q=\bar\chi^6
\end{equation}
where $\hat\Theta$ is the quadratic from corresponding to the $B$--field $B=2\pi \hat\Theta$.

\begin{thm}\cite{kkwk}
The algebra $\BTheta$ is the full matrix algebra  of $M_2(\T^2_{\theta})$ except in the following finite list of cases
\begin{enumerate}
\item $q=1$.
\item $q=-1$ and $\chi^4=1$.
\end{enumerate}
\end{thm}
The precise algebras are given in \cite{kkwk}. We wish to point out that $q=\chi=1$ is the commutative
case and $q=-\chi=1$ is isomorphic to the commutative case, while the other cases give non--commutative proper subalgebras of $M_2(\T^2_{\theta})$.

\subsection{The primitive cubic (P) case, and other Bravais cases}
For the simple cubic lattice and any other Bravais lattice of rank $k$ (P is the rank 3 case):
if $\Theta\neq 0$ then $\BTheta$ is simply the noncommutative torus $\TTheta^k$ and if $\Theta=0$ then
this $\B_0$ is the $C^*$ algebra of $T^k$. There are no degenerate points.

In the commutative case the cover $X\to T^k$ is trivial and so is the line bundle of Eigenvectors.

The analysis of \cite{BE} of the quantum Hall effect however suggests that there is a non--trivial noncommutative line
bundle in the case of $k=2$ for {\em non--zero $B$--field}.
Furthermore, in this case there is a non--trivial bundle, not using the noncommutative geometry, but rather the Eigenfunctions constructed in \cite{TKNN} for the full Hilbert space $\H$. This is what is also considered in \cite{simon}. We will study this phenomenon in the gyroid and the other cases in the future.

\subsection{The Diamond (D) case}

\subsubsection{The commutative case}
In this case, we see that the algebra $\BTheta$ is a subalgebra of $M_2(C(T^3))$, where $C(T^3)$ is the $C^*$ algebra of complex functions on the torus $T^3$.

The space $X$ defined by $\B$ in the commutative case is a generically 2--fold cover of the 3--torus $T^3$ where
 the ramification locus $\Tdeg$ is along three circles on $T^3$ given by the equations $\phi_i=\pi, \phi_j\equiv \phi_k+\pi\; \mbox{mod}\; 2\pi$ with $\{i,j,k\}=\{1,2,3\}$.
 $\Tdeg\Xi^{-1}(0)$ is the inverse image ---of the characteristic map--- of the only singular points (the origin) of the miniversal unfolding of $A_1$. Thus
 the singularities are of type $A_1$ but they are not discrete, but rather
 pulled back to the entire $\Tdeg$, hence there are also no Dirac points.

One can show that $T_0=T^3\setminus\Tdeg$ contracts onto a 1--dimensional CW--complex and hence has $H^2(T_0)=0$. Thus there are no non--vanishing topological charges associated to this geometry and no stability.

Analogous to the honeycomb case there are Abelian enhanced symmetries with 1--dimensional isotypical components, which force the double degeneracy in view of the structure of the Hamiltonian.

\subsubsection{The non--commutative case}

In the non--commutative case, we express our results in terms of parameters $q_i$ and $\xi_i$ defined as follows: Set
$e_1=\frac{1}{4}(1,1,1), e_2=\frac{1}{4}(-1,-1,1), e_3=\frac{1}{4}(-1,1,-1)$ for $B=2\pi\Theta$ let
\begin{equation}
{\Theta} (-e_1,e_2)=\varphi_1\quad
{\Theta} (-e_1,e_3)=\varphi_2\quad
{\Theta} (e_2,e_3)=\varphi_3  \text{ and }
\chi_i=e^{i \varphi_i}\;\mbox{for}\; i=1,2,3
\end{equation}

There are three operators  $U,V,W$, given explicitly in \cite{kkwk2}, which span $\T^3_{\Theta}$ and have commutation relations
\begin{equation}
U V = q_1 V U \quad
UW  = q_2 WU\quad
VW = q_3 WV
\end{equation}
where the $q_i$ expressed  in terms of the $\chi_i$ are:
\begin{equation}
q_1=\bar{\chi_1}^2 \chi_2^2 \chi_3^2 \quad
q_2=\bar{\chi_1}^6 \bar{\chi_2}^2 \bar{\chi_3}^2 \quad
q_3=\bar{\chi_1}^2\bar{ \chi_2}^6 \chi_3^2
\end{equation}
Vice versa, fixing the values of the $q_i$ fixes the $\chi_i$ up to eighth roots of unity:
\begin{equation}
\chi_1^8=\bar{q}_1 \bar{q}_2 \quad
\chi_2^8=q_1 \bar{q}_3  \quad
\chi_3^8=q_1^2\bar{q}_2 {q}_3
\end{equation}
Other useful relations are $ q_2 \bar{q}_3= \bar{\chi}_1^4 \chi_2^4 \bar{\chi}_3^4$ and $q_2 q_3 =\bar{\chi}_1^8 \bar{\chi}_2^8$.
the algebra $\BTheta$ is the {\em full}  matrix algebra {\em except} in the following cases in which it is
a proper subalgebra.

\begin{enumerate}
\item $q_1=q_2=q_3=1$ (the special bosonic cases)  and one of the following is true:

\begin{enumerate}
\item All $\chi_i^2=1$ then $\BTheta$ is isomorphic to the commutative algebra in the case of no magnetic field above.

\item Two of the $\chi_i^4=-1$, the third one necessarily being equal to $1$.\end{enumerate}
\item If $q_i=-1$ (special fermionic cases)  and  $\chi_i^4=1$. This means that either
\begin{enumerate}
\item all $\chi_i^2=-1$ or
\item  only one of the $\chi_i^2=-1$ the other two being $1$.
\end{enumerate}
\item $\bar q_1=q_2=q_3=\bar \chi^4_2$ and $\chi^2_1=1$ it follows that $\chi_2^4=\chi_3^4$. This is a one parameter family.
\item $q_1=q_2=q_3=\bar\chi_1^4$ and $\chi_2^2=1$ it follows that $ \chi_1^4=\bar \chi_3^4$. This is a one parameter family.
\item $q_1=q_2=\bar q_3=\bar \chi_1^4$ and $\chi_1^2=\bar\chi_2^2$. It follows that $\chi_3^4=1$. This is a one parameter family.

\end{enumerate}

\subsection{The Gyroid (G)  case}
\subsubsection{The commutative case}

For the gyroid, the commutative geometry if given by a generically unramified 4-fold cover of the three torus, see \cite{kkwk}.
There are only 4 ramification points.
This means that the locus is of real codimension 3 contrary to the D case where it was of codimension 2.
Furthermore the degenerations are 3 branches coming together at 2 points ---$(0,0,0)$ and $(\pi,\pi,\pi)$---
 and 2 pairs of branches coming
together at the other two points ---$(\frac{\pi}{2},\frac{\pi}{2},\frac{\pi}{2})$ and $(\frac{3\pi}{2},\frac{3\pi}{2},\frac{3\pi}{2})$. The latter furnish double Dirac points.

Using the characteristic map the first type of singular point corresponds to an $A_2$ singularity and the second type corresponds to the type $(A_1,A_1)$ stratum of the swallowtail. All the inverse images have discrete fibers. There are two image points on the $A_2$ stratum each with one inverse image under $\Xi$ and there is one point on the $(A_1,A_1)$ stratum, with two inverse images.

All the $A_1$ singularities in the fibers are Dirac points. That is there are four of these points.
Furthermore at all points there are enhanced symmetries by non--Abelian groups.

At $(0,0,0)$ the enhanced symmetry group is the symmetric group $\SS_4$ ---the full symmetry group of $\bG$ which entirely lifts to $Aut(T^3)$--- yielding one 1--dim irrep and one 3--dim irrep which forces the triple degeneracy. At $(\pi,\pi,\pi)$ we have an {\it a priori} projective
representation of $\SS_4$, which we showed however to be equivalent to the standard representation of $\SS_4$ and hence we again get one 1--dim irrep and one 3--dim irrep which forces the triple degeneracy.
At the other two points things are really interesting. The stabilizer symmetry group is $A_4$ and it yields a projective representation which is carried by the double cover of $A_4$ aka. $2A_4$, 2T, the binary tetrahedral group or $SL(2,3)$. The representation decomposes into two 2--dim irreps forcing the two double degeneracies.

Notice that we essentially need a projective representation, since $A_4$ itself has no 2--dim irreps.

Now $\Tdeg$ is the set of the four points above and $T_0=T^3\setminus \Tdeg$ contracts onto a 2--dim CW complex with non--trivial second homology.

Thus there are $K$--theoretic and cohomological charges. This is the special case of dimension 3 with codimension 3 degenerate points and moreover we have a slicing of $T^3$ by the fiber bundle $T^3\to T^2$ by any of the tree coordinate projections.  In fact the homology is generated by  any four slices which sit in between the 4 slices that contain the degenerate points. Pairing with these surfaces completely determines the Chern class of the line bundles and hence the line bundles up to isomorphism.

The relevant numerics were carries out in \cite{kkwk5}.
In accordance with the analytic calculations of \cite{simon, grove} the Dirac points yield jumps in the charge by $\pm 1$ for the two bands that cross.

A new result is that the $A_2$ points yield jumps by $-2,0,2$ for the three bands that cross.

All these charges are topologically stable. Again an interesting note is that the $A_2$ points
each split into four $A_1$ points in compliance with the jumps given above.

\subsubsection{The non--commutative case}

To state the results of \cite{kkwk} we use the bcc lattice vectors
\begin{equation}
\label{bccveceq2}
g_1=\frac{1}{2}(1,-1,1),\quad g_2=\frac{1}{2}(-1,1,1), \quad g_3=\frac{1}{2}(1,1,-1)\end{equation}

\begin{equation*}
 \theta_{12}=\frac{1}{2\pi}B\cdot (g_1\times g_2), \quad \theta_{13}=\frac{1}{2\pi}
 B\cdot (g_1\times g_3), \quad
\theta_{23}=\frac{1}{2\pi}B\cdot (g_2\times g_3)
\end{equation*}

\begin{equation*}
\alpha_1:=e^{2\pi i \theta_{12}}
\bar \alpha_2:= e^{2\pi i \theta_{13}}
\alpha_3:=e^{2\pi i \theta_{23}}
\end{equation*}

$$\phi_1=e^{\frac{\pi}{2} i \theta_{12}}, \quad \phi_2= e^{\frac{\pi}{2} i \theta_{31}},\quad
\phi_3= e^{\frac{\pi}{2} i \theta_{23}}, \quad \Phi=\phi_1\phi_2\phi_3
$$

{\bf Classification Theorem.}
\begin{enumerate}
\item If $\Phi\neq1$ or $\Phi=1$ and at least one $\alpha_i\neq 1$ and all $\phi_i$ are different then
$\BTheta=M_4(\T^3_{\Theta})$.

\item If $\phi_i=1$ for all $i$ then the algebra is the same as in the commutative case.

\item In all other cases $\B$ is non--commutative and $\BTheta\subsetneq M_4(\T^3_{\Theta})$.
\end{enumerate}

\subsection{Observation and conjecture}
Looking at the cases above, we observe several regularities. First
and foremost, there is agreement on the maximal dimension of the degenerate locus in $T^k$ between the commutative and the non--commutative case.  In the commutative case,
this locus is $\Tdeg$; in the non--commutative case, it is the values of the $B$--field,
which is again parameterized by $T^k$, now via $\Theta$, where the matrix algebra is not
the full matrix algebra.

We conjecture that this is always the case.

There are several possible points of attack here. The first is through the symmetries:
as we have seen, the re--gauging groupoid exists already in the non--commutative case.
Another is to consider how, in the presence of a conserved topological charge, larger
representations, such as $A_2$ in the gyroid case, break into smaller pieces.
Using the slicing method described above, one can readily see how that happens under a
deformation of the Hamiltonian in the commutative case. The question is whether the effect
of non-commutativity is something similar.

\section*{Acknowledgments}
RK thankfully acknowledges
support from NSF DMS-0805881.
BK  thankfully acknowledges support from the  NSF under the grant PHY-0969689.
  Any opinions, findings and conclusions or
recommendations expressed in this
 material are those of the authors and do not necessarily
reflect the views of the National Science Foundation.
The authors also thank the organizers of the conference: ``Noncommutative Algebraic Geometry and its Applications to Physics'' in Leiden for an event that catalyzed  new results and research directions. They also thank M.\ Marcolli for insightful discussions.

Parts of this work were completed when RK was visiting
the IHES in Bures--sur--Yvette, the Max--Planck--Institute in Bonn and the University of Hamburg with a Humboldt fellowship. He gratefully acknowledges
their support.

\end{document}